\pdfoutput=1
\RequirePackage{ifpdf}
\ifpdf 
\documentclass[pdftex]{sigma}
\else
\documentclass{sigma}
\fi

\numberwithin{equation}{section}

\newtheorem{Theorem}{Theorem}[section]

{ \theoremstyle{definition}
\newtheorem{Remark}[Theorem]{Remark}}

\newcommand{\I}{\scriptscriptstyle \rm I}
\newcommand{\II}{\scriptscriptstyle \rm I\hspace{-1pt}I}

\newcommand{\is}{\scriptscriptstyle \rm i}
\newcommand{\ii}{\scriptscriptstyle \rm ii}

\newcommand{\B}{\mathbb{B}}
\newcommand{\C}{\mathbb{C}}
\newcommand{\T}{\mathbb{T}}
\newcommand{\bu}{\bar u}
\newcommand{\bv}{\bar v}

\newcommand{\wt}[1]{\widetilde{#1}}

\newcommand{\grad}{\mathop{\rm grad}}

\begin{document}

\allowdisplaybreaks

\newcommand{\arXivNumber}{1611.00943}

\renewcommand{\PaperNumber}{015}

\FirstPageHeading

\ShortArticleName{Bethe Vectors for Composite Models with $\mathfrak{gl}(2|1)$ and $\mathfrak{gl}(1|2)$ Supersymmetry}

\ArticleName{Bethe Vectors for Composite Models \\ with $\boldsymbol{\mathfrak{gl}(2|1)}$ and $\boldsymbol{\mathfrak{gl}(1|2)}$ Supersymmetry}

\Author{Jan FUKSA~$^{\dag\ddag}$}

\AuthorNameForHeading{J.~Fuksa}

\Address{$^\dag$~Bogoliubov Laboratory of Theoretical Physics, Joint Institute for Nuclear Research,\\
\hphantom{$^\dag$}~Dubna, Russia}
\EmailD{\href{mailto:fuksa@theor.jinr.ru}{fuksa@theor.jinr.ru}}

\Address{$^\ddag$~Faculty of Nuclear Sciences and Physical Engineering, Czech Technical University in Prague,\\
\hphantom{$^\ddag$}~Czech Republic}

\ArticleDates{Received November 18, 2016, in f\/inal form March 03, 2017; Published online March 13, 2017}

\Abstract{Supersymmetric composite generalized quantum integrable models solvable by the algebraic Bethe ansatz are studied. Using a coproduct in the bialgebra of monodromy matrix elements and their action on Bethe vectors, formulas for Bethe vectors in the composite models with supersymmetry based on the super-Yangians $Y[\mathfrak{gl}(2|1)]$ and $Y[\mathfrak{gl}(1|2)]$ are derived.}


\Classification{17B37; 81R50; 82B23}

\section{Introduction}

The main success of the algebraic Bethe ansatz resides in the general prescription how to obtain eigenvectors (famous Bethe vectors) for a vast class of quantum integrable models. Consequently, it enables the calculation of the correlation functions via the calculation of the form factors. The Bethe vectors for models related to the algebra $\mathfrak{gl}(2)$ and its deformations have a very simple form (see \cite{Fad96,IK84,KBI93,FST80, TF79} and references therein). However, they become very nontrivial for models based on higher rank (super)algebras.

What is common for the Bethe vectors in the models with $\mathfrak{gl}(2)$ and higher rank symmetry (super)algebras is that they belong to the state (super)space $\mathcal{H}$ with a structure of the Fock space. In other words, $\mathcal{H}$ contains a cyclic vector $\Omega$ (usually called pseudovacuum) and the Bethe vectors are generated by the action of certain creation-like operators on $\Omega$. These creation-like operators belong to the (super)algebra of matrix elements of the monodromy matrix $T(u)$.

The problem of the higher rank (super)algebras has been addressed by the nested Bethe ansatz method. Its origin dates back to the works~\cite{Suth68, Yang67} in the context of the coordinate Bethe ansatz, whereas its algebraic version was introduced in~\cite{KuRe83}. The last years have been marked by considerable progress in f\/inding ``user-friendly'' forms of Bethe vectors for higher rank (super)algebras. For $\mathfrak{gl}(3)$ and its quantum deformations, Bethe vectors were calculated in~\cite{BPRS13BVgl3,BPRS13BVgl3q}. Bethe vectors for the superalgebras $\mathfrak{gl}(2|1)$ and $\mathfrak{gl}(1|2)$ were obtained in~\cite{PRS16bethe} and used to calculate the scalar products~\cite{HLPRS16scalar1,HLPRS16scalar2} and the form factors of the monodromy matrix elements~\cite{HLPRS16FF}. Based on this, correlation functions in models with this type of supersymmetry can be investigated via the form factor expansion.

The knowledge of form factors of local operators thus becomes an important condition for the calculation of correlation functions. In their turn, the form factors of local operators can be computed in two dif\/ferent ways. To clarify the purpose and content of this paper, we brief\/ly describe these approaches.

The f\/irst method is applied to the models with a known solution of the quantum inverse scattering problem (see, e.g.,~\cite{KMT00,MT00}).
These include the supersymmetric $t$-$J$ model (in the context of $\mathfrak{gl}(2|1)$) well-known from the condensed matter physics. The corresponding quantum inverse scattering problem was solved in~\cite{GoKo00}. The form factors of local operators in this model reduce to the form factors of the monodromy matrix elements. Their calculation, in its turn, reduces to the well-studied problem of calculating scalar products of Bethe vectors. Thus, in this case, the inverse scattering problem provides a powerful tool for calculating the form factors of local operators. However, solutions of the quantum inverse scattering problem are known for a rather restricted class of integrable models.

The second approach to the calculation of the form factors of local operators is more general and does not imply knowledge of explicit solutions of the quantum inverse scattering problem. It is based on the composite model~\cite{IK84}.\footnote{The terminology witnessed a long evolution here. The composite model was originally called the two-site model in~\cite{IK84}. Later the term two-component model was used~\cite{Slav07}. Both these terms can lead to confusion in some situations. Therefore, the term composite model was proposed recently~\cite{PRS15gl3compBV}. We hold this terminology in our article.} There are several integrable models related to the~$\mathfrak{gl}(2|1)$ superalgebra to which the composite model is applicable. They are, e.g., the impurity version of the $t$-$J$ model~\cite{FLT99}, the doped Heisenberg chains \cite{Frahm99}, or the supersymmetric $U$ model~\cite{Bracken95} and~\cite{Bedurftig95,Pfannmuller96}. Unlike the $t$-$J$ model, they are not based on the fundamental representation of~$\mathfrak{gl}(2|1)$, therefore, the quantum inverse scattering problem is unsolved for them. However, the method of the composite model still works in this case, as it is based on the algebraic Bethe ansatz only.

The main idea of the composite model method is that the interval $[0,L]$, on which the original model is def\/ined, is divided into two subintervals $[0,x]$ and $]x,L]$. Consequently, the state (super)space $\mathcal{H}$ of the complete model is divided into two (graded) subspaces~$\mathcal{H}^{(1)}$ and~$\mathcal{H}^{(2)}$ corresponding to $[0,x]$ and $]x,L]$, respectively, such that $\mathcal{H}=\mathcal{H}^{(1)}\otimes\mathcal{H}^{(2)}$. Simultaneously, the monodromy matrix $T(u)$ can be presented as a matrix product of two partial monodromy matrices $T^{(2)}(u)$ and $T^{(1)}(u)$:
\begin{gather}\label{monP}
T(u) = T^{(2)}(u) \cdot T^{(1)}(u).
\end{gather}
The partial monodromy matrices $T^{(\ell)}(u)$, $\ell=1,2$, are to be understood as particular representations of the complete monodromy matrix $T(u)$ on the partial state (super)spaces~$\mathcal{H}^{(\ell)}$. This shall be ascribed to a bialgebraic structure on the algebra of matrix elements of the monodromy matrix. The property~\eqref{monP} is equivalent to a comultiplication in the aforementioned bialgebra. The partial monodromy matrix~$T^{(\ell)}(u)$ acts nontrivially only in its corresponding representation (super)space $\mathcal{H}^{(\ell)}$.

We suppose that the partial state (super)spaces $\mathcal{H}^{(\ell)}$ possess also the structure of the Fock space with the partial pseudovacua $\Omega^{(\ell)}$ such that $\Omega=\Omega^{(1)} \otimes \Omega^{(2)}$. A nontrivial fact is that the Bethe vectors $\B\in\mathcal{H}$ of the complete model can be represented as bilinear combinations of partial Bethe vectors $\B^{(1)}\in\mathcal{H}^{(1)}$ and $\B^{(2)}\in\mathcal{H}^{(2)}$. For the models based on the $\mathfrak{gl}(2)$ symmetry, this was shown in~\cite{IK84}. The case of~$\mathfrak{gl}(3)$ was recently studied in~\cite{PRS15gl3compBV}. The aim of this article is to f\/ind a~similar representation for the integrable models based on the superalgebras $\mathfrak{gl}(2|1)$ and~$\mathfrak{gl}(1|2)$.
Such a representation allows one to compute the form factors of the partial monodromy matrix elements $T_{ij}^{(\ell)}(u)$ in the basis of the Bethe vectors of the complete model. Based on this, the form factors and correlation functions of local operators can be investigated.

It is worth mentioning that the composite model is tied with the theory of the quantized Knizhnik--Zamolodchikov (qKZ) equation. Bethe vectors play a principal role in f\/inding its integral solutions \cite{BKZ97, TV95}. This is related to the fact that the Bethe vectors belong to the same vector space as the solutions of qKZ. For an $R$-matrix~$R(u)$, which at a certain value of the parameter $u$ becomes a permutation operator, a technique of solving the qKZ based on the nested Bethe ansatz was developed in~\cite{BKZ97}. In~\cite{TV95}, the so-called weight functions are constructed as basic building blocks for the integral solutions of the qKZ. They satisfy a coproduct property. The coproduct property states that from two weight functions belonging to vector spaces~$V^{(1)}$ and $V^{(2)}$, respectively, a weight function belonging to their tensor product $V^{(1)}\otimes V^{(2)}$ can be constructed in a unique way. The weight functions are nothing else but the Bethe vectors for an integrable model whose monodromy matrix is constructed from the $R$-matrices corresponding to the qKZ equation. The composite model represents in an explicit way the coproduct property of the weight functions, i.e., of the Bethe vectors.

The general strategy of this article is similar to the strategy of the paper \cite{PRS15gl3compBV} on $\mathfrak{gl}(3)$. However, some technical details appear due to dif\/ferent commutation relations.
We describe these technical dif\/ferences. We show how they change calculations for a concrete case, but we do not give all calculations.

Recall that the superalgebras $\mathfrak{gl}(2|1)$ and $\mathfrak{gl}(1|2)$ are isomorphic. The same holds also for their corresponding super-Yangians $Y[\mathfrak{gl}(2|1)]$ and $Y[\mathfrak{gl}(1|2)]$. Therefore, we mostly concern ourselves with the models based on $\mathfrak{gl}(2|1)$ and then isomorphically map the results to the $\mathfrak{gl}(1|2)$ case. Let us emphasize that the underlying algebraic structure of the models described in this paper are the super-Yangians $Y[\mathfrak{gl}(2|1)]$ and $Y[\mathfrak{gl}(1|2)]$, respectively.

The paper is organized as follows. In Section~\ref{sec:gl21basic}, basic notions of the $\mathfrak{gl}(2|1)$-invariant quantum integrable models like the monodromy matrix, the RTT algebra, and the Bethe vectors are introduced, and the notation used further in the article is described. Section \ref{sec:gl21composite} describes the $\mathfrak{gl}(2|1)$-invariant composite model. The main results are formulated in two theorems contained in Section~\ref{sec:Results21}. Section~\ref{sec:Action} deals with some technical details from Section~\ref{sec:Results21}. Section~\ref{sec:gl12basics} is devoted to the $\mathfrak{gl}(1|2)$-invariant composite model. The main results are presented in two theorems in Section~\ref{sec:Results12}.

\section[Basic notions of the $\mathfrak{gl}(2|1)$-invariant model]{Basic notions of the $\boldsymbol{\mathfrak{gl}(2|1)}$-invariant model}\label{sec:gl21basic}

\subsection[$R$-matrix and RTT algebra]{$\boldsymbol{R}$-matrix and RTT algebra} \label{sec:RTT}

The content of this subsection can be easily generalized to superalgebras of other ranks.

The $R$-matrix acts in the tensor product of two $\mathbb{Z}_2$-graded auxiliary superspaces $\mathbb{C}^{2|1}$. The basis of the even part of $\mathbb{C}^{2|1}$ is $\{e_1,e_2\}$ and of the odd part is $\{e_{3}\}$. We introduce the parity function $[\ ]\colon \{1,2,3 \}\rightarrow \mathbb{Z}_2$ such that $[1]=[2]=0$ and $[3]=1$. The gradation of the superspace~$\mathbb{C}^{2|1}$ is thus described by this parity function: $\grad(e_i)=[i]$, $i=1,2,3$.

The matrix units $E_{ij}\in \mathrm{End}(\mathbb{C}^{3})$ are introduced in the standard way $ (E_{ij})_{ab} = \delta_{ia} \delta_{jb}$ with the gradation
$ \grad(E_{ij}) = [i]+ [j]$. The tensor product is graded in the following way:
\begin{gather*}
(E_{ij} \otimes E_{kl}) (E_{mn}\otimes E_{pq}) = (-1)^{([k]+[l])([m]+[n])} E_{ij} E_{mn} \otimes E_{kl} E_{pq}.
\end{gather*}
The graded permutation (superpermutation) for two auxiliary superspaces \cite{KS82solutions} is def\/ined using the matrix units
\begin{gather*}
 P = \sum_{i,j=1}^{3} (-1)^{[j]} E_{ij} \otimes E_{ji}.
\end{gather*}

The $\mathfrak{gl}(2|1)$-invariant $R$-matrix has the explicit form
\begin{gather}\label{Rmat}
R(u,v) = \mathbb{I} \otimes \mathbb{I} + g(u,v) P,
\end{gather}
where $\mathbb{I}$ is the unit matrix in $\mathbb{C}^{2|1}$. The function $g(u,v)$ is antisymmetric and rational{\samepage
\begin{gather*}
g(u,v) = \frac{c}{u-v}
\end{gather*}
and $c\in\mathbb{C}$ is an auxiliary constant.}

The def\/inition of the monodromy matrix is standard $T(u) = \sum\limits_{i,j=1}^{3} E_{ij} \otimes T_{ij}(u)$. The ele\-ments $T_{ij}(u)$, $i,j=1,2,3$, together with the unit element $\mathbf{1}$ are generators of the associative superalgebra~$\mathcal{A}$. The monodromy matrix is a globally even matrix due to the fact that $\grad(E_{ij})=\grad(T_{ij}(u))$. It satisf\/ies the RTT relation with the $R$-matrix~\eqref{Rmat}
\begin{gather}\label{RTT}
R(u,v) (T(u) \otimes \mathbb{I}) (\mathbb{I}\otimes T(v)) = (\mathbb{I}\otimes T(v)) (T(u) \otimes \mathbb{I}) R(u,v).
\end{gather}
The RTT relation is equivalent to the bilinear relations in~$\mathcal{A}$ which can be written in two equivalent forms
\begin{gather}
[ T_{ij}(u), T_{kl}(v)\} = g(u,v) (-)^{[i][j]+[i][l]+[j][l]} (T_{il}(u)T_{kj}(v) - T_{il}(v)T_{kj}(u)) \notag\\
\hphantom{[ T_{ij}(u), T_{kl}(v)\} }{} = - g(u,v) (-)^{[i][k]+[i][l]+[k][l]} (T_{kj}(u)T_{il}(v) - T_{kj}(v)T_{il}(u)) \label{Product}
\end{gather}
with the supercommutator $ [ T_{ij}(u), T_{kl}(v)\} \equiv T_{ij}(u) T_{kl}(v) - (-)^{([i]+[j])([k]+[l])} T_{kl}(v) T_{ij}(u)$. The superalgebra~$\mathcal{A}$ is called the RTT algebra.

There is a coalgebraic structure def\/ined on $\mathcal{A}$ with the following coproduct:
\begin{gather}
\Delta(T_{ij}(u)) \equiv \sum_{k=1}^3 T_{kj}(u) \otimes T_{ik}(u) = \sum_{k=1}^3 T_{kj}^{(1)} (u) T_{ik}^{(2)} (u). \label{Coproduct}
\end{gather}
The superscripts at $T_{ij}^{(\ell)}(u)$, $\ell=1,2$, are used to distinguish the two copies of the bialgebra~$\mathcal{A}$. The elements of dif\/ferent copies of~$\mathcal{A}$ mutually supercommute $ \big[ T_{il}^{(1)} (u),T_{jk}^{(2)} (v) \big\} = 0$. The partial monodromy matrices $T^{(\ell)}(u)$, $\ell=1,2$, satisfy the same RTT relation~\eqref{RTT} as $\mathcal{A}$ has the structure of bialgebra. It is worth mentioning that the coproduct~\eqref{Coproduct} is equivalent to relation~\eqref{monP} for the monodromy matrix of the composite model.

Let us remind that there is an antimorphism of $\mathcal{A}$ \cite{PRS16bethe}:
\begin{gather}\label{amorph1}
 \psi( T_{ij}(u) ) = (-1)^{[i][j]+[i]} T_{ji}(u),\qquad
\psi(A B) = (-1)^{\grad(A)\cdot\grad(B)} \psi(B) \psi(A), 
\end{gather}
for $A,B\in \mathcal{A}$ of def\/inite gradings. It preserves the supercommutator $ \psi([A,B\}) = -[\psi(A),\psi(B)\}$ and satisf\/ies the composition rule
\begin{gather}\label{keyCP}
\Delta \circ \psi = (\psi \otimes \psi) \circ \Delta'
\end{gather}
with the standard coproduct \eqref{Coproduct} and the opposite coproduct\footnote{The name {\em opposite coproduct} is obviously relative. It refers to the fact that it intertwines the factors in the tensor product in comparison with the ``standard'' coproduct~\eqref{Coproduct}. In the same way, \eqref{Coproduct} is the opposite coproduct to \eqref{OpCopr}. Despite this, we call in this article the coproduct of the type \eqref{Coproduct} the {\em standard} and of \eqref{OpCopr} the {\em opposite}.}
\begin{gather}\label{OpCopr}
\Delta'(T_{ij}(u)) = \sum_k (-1)^{([i]+[k])([k]+[j])} T_{ik}(u) \otimes T_{kj}(u) = \sum_k T^{(2)}_{kj} (u) T^{(1)}_{ik} (u).
\end{gather}
The composition rule~\eqref{keyCP} is the key property for the investigation of dual Bethe vectors in the composite model, as we will see below.

\subsection{Notation}

The following functions are used throughout the text:
\begin{gather*}
 f(u,v) = 1+g(u,v) = \frac{u-v+c}{u-v}, \qquad h(u,v) = \frac{f(u,v)}{g(u,v)} = \frac{u-v+c}{c}.
\end{gather*}

We completely follow the notation used, e.g., in~\cite{PRS16bethe}. The sets of parameters are denoted as~$\bu$, $\bv$ etc. Their individual elements as $u_j$,~$v_0,$ etc. The notation~$\bu_j$, $\bv_0$ means $\bu\backslash u_j$, $\bv\backslash v_0$, etc. To avoid lengthy and complicated formulas, we use the shorthand notation for products of the above functions over the sets of parameters. For example,
\begin{gather*}
f(\bu,v) \equiv \prod_{u_j\in\bu} f(u_j,v), \qquad g(u,\bv_i)\equiv \prod_{\substack{v_j\in\bv \\ v_j\neq v_i}} g(u,v_j), \qquad h(\bu,\bv) = \prod_{u_j\in\bu} \prod_{v_k\in\bv} h(u_j,v_k).
\end{gather*}
This notation is preserved also for products of even operators
\begin{gather*}
 T_{ij}(\bu) \equiv \prod_{u_k\in\bu} T_{ij}(u_k), \qquad [i]+[j]=0.
\end{gather*}

For products of odd operators $T_{i3}(u)$ and $T_{3i}(u)$, $i=1,2$, the symmetrised product is used
\begin{gather*}
\mathbb{T}_{i3}(\bar u) \equiv \frac{T_{i3}(u_1)T_{i3}(u_2)\cdots T_{i3}(u_n)}{\prod\limits_{1\leq j<k \leq n} h(u_k,u_j)}, \qquad \mathbb{T}_{3i}(\bar u) \equiv \frac{T_{3i}(u_1)T_{3i}(u_2)\cdots T_{3i}(u_n)}{\prod\limits_{1\leq j<k \leq n} h(u_j,u_k)}.
\end{gather*}

A set of parameters $\bu$ is often divided into its two disjoint subsets $\bu_{\I}$, $\bu_{\II}$: $\bu_{\I} \cap \bu_{\II} =\varnothing$ and $\bu_{\I} \cup \bu_{\II} = \bu$. We denote it as $\bu\Rightarrow \{\bu_{\I},\bu_{\II}\}$. Such a partition is always accompanied by a summation over all partitions of the prescribed type, just according to the rule: ``{\it where a~partition is, there is a summation}''.

\subsection{Bethe vectors}

The Bethe vectors for the $\mathfrak{gl}(2|1)$-invariant models were calculated in~\cite{PRS16bethe}. They are obtained under the assumption that there is a cyclic vector~$\Omega$ which is an eigenvector of the diagonal elements of the monodromy matrix $T(u)$ and is annihilated by the lower-triangular elements:
\begin{gather}\label{ass1}
 T_{ii}(u) \Omega = \lambda_i(u) \Omega, \qquad T_{ij}(u)\Omega = 0, \quad i>j.
\end{gather}
The eigenfunctions $\lambda_{i}(u)$ vary depending on the model. Leaving them as arbitrary functions of the parameter~$u$, the model is called generalized. The cyclic vector $\Omega$ is called pseudovacuum. We suppose its gradation to be vanishing ($\Omega$~is an even supervector). The other Bethe vectors $\B_{a,b}(\bu;\bv)$ are generated by the action of the upper-triangular elements $T_{ij}(u)$, $i<j$, on~$\Omega$. They depend on two sets of spectral parameters $\bu=\{u_1,\dots,u_a\}$ and $\bv=\{v_1,\dots,v_b\}$ where $a,b=0,1,2,\dots$ are the cardinalities $a=\# \bu$ and $b=\#\bv$.

The Bethe vectors have several representations \cite{PRS16bethe}. It is convenient to use the following one for our purposes:
\begin{gather}\label{BetheVec}
\mathbb{B}_{a,b}(\bu ; \bv) = \sum K_n(\bv_{\I} | \bu_{\I}) \frac{f(\bu_{\I} , \bu_{\II}) g(\bv_{\II}, \bv_{\I})}{\lambda_2(\bu_{\II}) \lambda_2(\bv) f(\bv, \bu)} \mathbb{T}_{13}(\bv_{\I}) \mathbb{T}_{23}(\bv_{\II}) T_{12}(\bu_{\II}) \Omega.
\end{gather}
The sum goes over all partitions $\bu \Rightarrow \{\bu_{\I}, \bu_{\II} \}$ and $\bv \Rightarrow \{\bv_{\I}, \bv_{\II} \}$ with the restriction that $\# \bu_{\I} = \# \bv_{\I} = n$, where $n=0,1,\dots,\min (a,b)$. The function $K_n(\bv_{\I} | \bu_{\I})$ is the partition function of the six-vertex model with the domain wall boundary condition~\cite{Kor82}. It has the following representation~\cite{Iz87}:
\begin{gather*}
K_n(\bv|\bu) = \prod_{i<j}^n g(v_i,v_j) g(u_j,u_i) \cdot \frac{f(\bv,\bu)}{g(\bv,\bu)} \left. \det\left[ \frac{g^2(v_k,u_l)}{f(v_k,u_l)} \right] \right|_{k,l=1,\dots,n}.
\end{gather*}
 A simple observation states that the Bethe vector $\B_{a,b}(\bu;\bv)$ has the gradation $\grad(\B_{a,b})=b$ $(\mathop{\rm mod} 2)$.

Similarly, we assume the existence of the dual pseudovacuum $\Omega^\dagger$ with the properties
\begin{gather*}
 \Omega^\dagger T_{ii}(u) = \lambda_i(u) \Omega^\dagger, \qquad\Omega^\dagger T_{ij}(u)=0, \qquad i<j,
\end{gather*}
where the eigenfunctions $\lambda_i(u)$ are the same as in~\eqref{ass1}. We suppose the gradation of $\Omega^\dagger$ to be vanishing. The dual Bethe vectors $\C_{ab}(\bu;\bv)$ also depend on two sets of spectral parame\-ters~$\bu$,~$\bv$ with $a=\#\bu$, $b=\#\bv$. Their explicit form is~\cite{PRS16bethe}
\begin{gather*}
\C_{a,b}(\bar u; \bar v) = (-1)^{\frac{b^2-b}{2}} \sum K_n(\bar v_{\I} | \bar u_{\I}) \frac{f(\bar u_{\I} , \bar u_{\II}) g(\bar v_{\II}, \bar v_{\I})}{\lambda_2(\bar u_{\II}) \lambda_2(\bar v) f(\bar v, \bar u)} \Omega^\dagger T_{21}(\bar u_{\II}) \T_{32}(\bar v_{\II}) \T_{31}(\bar v_{\I}),
\end{gather*}
where the sum goes over all partitions $\bar u \Rightarrow \{\bar u_{\I}, \bar u_{\II} \}$ and $\bar v \Rightarrow \{\bar v_{\I}, \bar v_{\II} \}$ with the restriction that $\# \bar u_{\I} = \# \bar v_{\I} = n$, where $n=0,1,\dots,\min (a,b)$. The gradation of the dual Bethe vector $\C_{a,b}(\bu;\bv)$ is again $\grad(\C_{a,b})=b$ $(\mathop{\rm mod} 2)$.

If the antimorphism \eqref{amorph1} relates the pseudovacuum to the dual pseudovacuum as $\psi(\Omega) = \Omega^\dagger$, the Bethe vector $\B_{a,b}(\bar u;\bar v)$ is mapped by $\psi$ to the dual Bethe vector $\C_{a,b}(\bar u;\bar v)$
\begin{gather} \label{psiBC}
\psi\left(\B_{a,b}(\bar u;\bar v)\right) =\C_{a,b}(\bar u;\bar v).
\end{gather}

\section[Composite $\mathfrak{gl}(2|1)$-invariant model]{Composite $\boldsymbol{\mathfrak{gl}(2|1)}$-invariant model} \label{sec:gl21composite}

The interval $[0,L]$, which the generalized model is def\/ined on, is split into its two subintervals, as discussed in Introduction. The fact that the monodromy matrix of the full model $T(u)$ is simply the matrix product of the partial monodromy matrices $T^{(2)} (u) T^{(1)} (u)$, as expressed in~\eqref{monP}, follows from the coproduct~\eqref{Coproduct} in the bialgebra~$\mathcal{A}$.

This is followed by the split of the original superspace into its two graded subspaces $\mathcal{H}=\mathcal{H}^{(1)}\otimes \mathcal{H}^{(2)}$. We suppose that the graded subspaces $\mathcal{H}^{(\ell)}$, $\ell=1,2$, contain the partial pseudovacua~$\Omega^{(\ell)}$ with the properties
\begin{gather*}
 T_{ii}^{(\ell)}(u) \Omega^{(\ell)} = \lambda_i^{(\ell)}(u) \Omega^{(\ell)}, \qquad T_{ij}^{(\ell)}(u) \Omega^{(\ell)}=0, \qquad i> j.
\end{gather*}
The operators $T_{ij}^{(\ell)}(u)$ are the matrix elements of the partial monodromy matrices $T^{(\ell)}(u)$ described by \eqref{monP}. In other words, the subspaces $\mathcal{H}^{(\ell)}$ have exactly the same structure of a Fock space as the full superspace $\mathcal{H}$. The corresponding Bethe vectors are denoted as $\B^{(\ell)}_{a,b}(\bu;\bv)$. We suppose that the partial pseudovacua $\Omega^{(\ell)}$ form the total pseudovacuum $\Omega=\Omega^{(1)}\Omega^{(2)}$. We omit here the symbol $\otimes$ for the direct product of the pseudovacua because the superscripts at~$\Omega^{(\ell)}$ indicate that we work in the direct product of superspaces. This notation is kept also below for direct products of arbitrary Bethe vectors.

Similarly, we suppose for the partial dual superspaces the existence of the dual pseudova\-cua~$\Omega^{\dagger(\ell)}$ satisfying
\begin{gather*}
 \Omega^{\dagger(\ell)} T_{ii}^{(\ell)}(u) = \lambda_i^{(\ell)}(u) \Omega^{\dagger(\ell)}, \qquad \Omega^{\dagger(\ell)} T_{ij}^{(\ell)}(u) = 0, \qquad i< j.
\end{gather*}
The corresponding dual Bethe vectors are denoted as $\C^{(\ell)}_{a,b}(\bu;\bv)$ and we again suppose that $\Omega^\dagger =\Omega^{\dagger(1)}\Omega^{\dagger(2)}$.

It is useful to introduce two ratio functions $r_1(u)$ and $r_3(u)$ instead of three independent functions~$\lambda_i(u)$:
\begin{gather*}
 r_i(u) = \frac{\lambda_i(u)}{\lambda_2(u)}, \quad i=1,3, \qquad r_i^{(\ell)}(u) = \frac{\lambda^{(\ell)}_i(u)}{\lambda^{(\ell)}_2(u)}, \qquad i=1,3, \quad \ell=1,2.
\end{gather*}
This corresponds to the multiplication of the monodromy matrix by $\lambda_2^{-1}(u)$. Obviously $\lambda_{i}(u) = \lambda_i^{(1)}(u) \lambda_i^{(2)}(u)$ and $r_i(u) = r_i^{(1)}(u) r_i^{(2)} (u)$.

\subsection{Main results} \label{sec:Results21}

\begin{Theorem} \label{theorB}
The Bethe vectors of the full model can be expressed as the bilinear combination of the partial Bethe vectors:
\begin{gather} \label{BVcomp}
\mathbb{B}_{a,b}(\bu;\bv) = \sum r_1^{(2)}(\bu_{\I}) r_3^{(1)}(\bv_{\II}) \frac{f(\bu_{\II}, \bu_{\I})g(\bv_{\I},\bv_{\II})}{f(\bv_{\II},\bu_{\I})} \mathbb{B}_{a_2,b_2}^{(2)}(\bu_{\II}; \bv_{\II}) \mathbb{B}^{(1)}_{a_1,b_1}(\bu_{\I}; \bv_{\I}).
\end{gather}
The summation goes over all partitions $\bar u \Rightarrow \{\bu_{\I}, \bu_{\II}\}$ and $\bv \Rightarrow \{\bv_{\I}, \bv_{\II}\}$ with no restriction. The corresponding cardinalities satisfy $a_1+a_2=a$ and $b_1+b_2=b$.
\end{Theorem}

The coproduct used in the proof is the standard one~\eqref{Coproduct}, as can be seen in Section~\ref{sec:Action}.

\begin{proof}[Idea of the proof] The proof is based on a recursion relation enjoined by Bethe vectors of $\mathfrak{gl}(2|1)$-invariant models~\cite{PRS16bethe}
\begin{gather}
\frac{T_{23}(z)}{\lambda_2(z)h(\bv,z)} \mathbb{B}_{a,b-1}(\bu;\bv) = f(z,\bu) \mathbb{B}_{a,b}(\bu; \{z,\bv\}) \nonumber\\
\hphantom{\frac{T_{23}(z)}{\lambda_2(z)h(\bv,z)} \mathbb{B}_{a,b-1}(\bu;\bv) =}{} + \sum_{\bu \Rightarrow \{u_0,\bu_0\}} g(u_0,z) f(u_0,\bu_0) \frac{T_{13}(z)\mathbb{B}_{a-1,b-1}(\bu_0;\bv)}{\lambda_2(z)h(\bv,z)}.\label{RecRelBV}
\end{gather}
Here, the sum is taken over all partitions $\bu\Rightarrow\{u_0,\bu_0\}$ where $\# u_0=1.$

Let us def\/ine the following vectors contained in the composite model, i.e., in $\mathcal{H}^{(1)}\otimes \mathcal{H}^{(2)}$:
\begin{gather}\label{Conjecture}
\mathcal{B}_{a,b}(\bu;\bv) = \sum r_1^{(2)}(\bu_{\I}) r_3^{(1)}(\bv_{\II}) \frac{f(\bu_{\II}, \bu_{\I})g(\bv_{\I},\bv_{\II})}{f(\bv_{\II},\bu_{\I})} \mathbb{B}^{(2)}(\bu_{\II}; \bv_{\II}) \mathbb{B}^{(1)}(\bu_{\I}; \bv_{\I}).
\end{gather}
The subscripts $a_1$, $b_1$, $a_2$, $b_2$ of the partial Bethe vectors $\B^{(1)}(\bu_{\I}; \bv_{\I})$ and $\B^{(2)}(\bu_{\II}; \bv_{\II})$ in the def\/inition of $\mathcal{B}_{a,b}(\bu;\bv)$ are omitted here and below because they do not carry any important information. Obviously, if we prove that vector~\eqref{Conjecture} satisf\/ies recursion \eqref{RecRelBV} and an initial condition $\mathcal{B}_{a,0}(\bu;\varnothing) =\mathbb{B}_{a,0}(\bu;\varnothing)$, $a=0,1,\dots$, then, using induction over $b$, we immediately obtain that $\mathcal{B}_{a,b}(\bu;\bv) =\mathbb{B}_{a,b}(\bu;\bv)$ for $a$ and $b$ arbitrary.

It is easy to see that $\mathcal{B}_{a,0}(\bu;\varnothing)$ coincides with the known result for the composite $\mathfrak{gl}(2)$-invariant model \cite{IK84}
\begin{gather*}
\mathcal{B}_{a,0}(\bu;\varnothing) = \sum r_1^{(2)}(\bu_{\I}) f(\bu_{\II},\bu_{\I}) \B^{(2)}(\bu_{\II};\varnothing) \B^{(1)}(\bu_{\I};\varnothing) = \B_{a,0}(\bu;\varnothing).
\end{gather*}
Thus, the initial condition is satisf\/ied. It remains to prove that $\mathcal{B}_{a,b}(\bu;\bv)$ satisf\/ies the recursion
\begin{gather}
\frac{T_{23}(z)}{\lambda_2(z)h(\bv,z)} \mathcal{B}_{a,b-1}(\bu;\bv) = f(z,\bu) \mathcal{B}_{a,b}(\bu; \{z,\bv\}) \nonumber\\
\hphantom{\frac{T_{23}(z)}{\lambda_2(z)h(\bv,z)} \mathcal{B}_{a,b-1}(\bu;\bv) =}{} + \sum_{\bu \Rightarrow \{u_0,\bu_0\}} g(u_0,z) f(u_0,\bu_0) \frac{T_{13}(z)\mathcal{B}_{a-1,b-1}(\bu_0;\bv)}{\lambda_2(z)h(\bv,z)}. \label{RecRel}
\end{gather}
Here the notation is the same as in~\eqref{RecRelBV}.

One can easily convince oneself that recursion relation~\eqref{RecRel} is a simple consequence of the following two relations which we intend to prove in Section~\ref{sec:Action}:
\begin{gather}\label{T13Bcal0}
\frac{T_{13}(z)}{\lambda_{2}(z)h(\bv,z)} \mathcal{B}_{a-1,b-1}(\bu;\bv) = \mathcal{B}_{a,b}(\{z,\bu\}; \{z,\bv\}),\\
\frac{T_{23}(z)}{\lambda_2(z)h(\bv,z)} \mathcal{B}_{a,b-1}(\bu;\bv) = f(z,\bu) \mathcal{B}_{a,b}(\bu; \{z,\bv\})\nonumber \\
\hphantom{\frac{T_{23}(z)}{\lambda_2(z)h(\bv,z)} \mathcal{B}_{a,b-1}(\bu;\bv) =}{} + \sum_{\bu\Rightarrow\{u_0,\bu_0\} } g(u_0,z) f(u_0,\bu_0) \mathcal{B}_{a,b}(\{z,\bu_0\};\{z,\bv\}).\label{T23Bcal}
\end{gather}
The sum is performed here over all partitions of the type $\bu\Rightarrow\{u_0,\bu_0\}$ where $\# u_0=1$. Thus, the proof of equations~\eqref{T13Bcal0},~\eqref{T23Bcal} yields the proof of~\eqref{BVcomp}.
\end{proof}

\begin{Remark} We can commute the factors $\B^{(1)}_{a_1,b_1}(\bu_{\I}; \bv_{\I})$ and $\B^{(2)}_{a_2,b_2}(\bu_{\II}; \bv_{\II})$ in \eqref{BVcomp} according to
\begin{gather*}
g(\bv_{\I},\bv_{\II}) \B_{a_2,b_2}^{(2)}(\bu_{\II}; \bv_{\II}) \B^{(1)}_{a_1,b_1}(\bu_{\I}; \bv_{\I}) = g(\bv_{\II},\bv_{\I}) \B^{(1)}_{a_1,b_1}(\bu_{\I}; \bv_{\I}) \B_{a_2,b_2}^{(2)}(\bu_{\II}; \bv_{\II})
\end{gather*}
because their gradation is ref\/lected in the product of the antisymmentric functions $g(\bv_{\I},\bv_{\II})$.
\end{Remark}

\begin{Theorem} \label{theorC}
The dual Bethe vectors of the full model can be expressed as the bilinear combination of the partial dual Bethe vectors:
\begin{gather*}
\C_{a,b}(\bu;\bv )= \sum r_1^{(1)}(\bu_{\II}) r_3^{(2)}(\bv_{\I}) \frac{f(\bu_{\I}, \bu_{\II}) g(\bv_{\II},\bv_{\I})}{f(\bv_{\I},\bu_{\II})} \C_{a_1,b_1}^{(1)}(\bu_{\I}; \bv_{\I}) \C^{(2)}_{a_2,b_2}(\bu_{\II}; \bv_{\II}).
\end{gather*}
The summation goes over all partitions $\bu \Rightarrow \{\bu_{\I}, \bu_{\II}\}$ and $\bv \Rightarrow \{\bv_{\I}, \bv_{\II}\}$. The corresponding cardinalities satisfy $a_1+a_2=a$ and $b_1+b_2=b$.
\end{Theorem}

\begin{proof}
The Bethe vector $\B_{a,b}(\bar u;\bar v)$ can be understood as
\begin{gather*}
\B_{a,b}(\bu;\bv) = B_{a,b}(\bu; \bv) \Omega
\end{gather*}
where $B_{a,b}(\bu; \bv)$ is a polynomial in elements of the bialgebra $\mathcal{A}$ which acts on the pseudova\-cuum~$\Omega$. In other words, it is the rest of the Bethe vector~\eqref{BetheVec} if we erase the pseudovacuum. Hence, formula \eqref{BVcomp} for Bethe vectors in the composite model can be written as
\begin{gather*}
\B_{a,b}(\bu;\bv) = B_{a,b}(\bu;\bv)\Omega = \Delta(B_{a,b}(\bu;\bv)) \Omega^{(1)} \Omega^{(2)} \\
\hphantom{\B_{a,b}(\bu;\bv)}{} = \sum r_1^{(2)}(\bu_{\I}) r_3^{(1)}(\bv_{\II}) \frac{f(\bu_{\II}, \bu_{\I})g(\bv_{\I},\bv_{\II})}{f(\bv_{\II},\bu_{\I})} B^{(2)}_{a_2,b_2}(\bu_{\II};\bv_{\II}) B^{(1)}_{a_1,b_1}(\bu_{\I};\bv_{\I}) \Omega^{(1)}\Omega^{(2)}.
\end{gather*}

The antimorphism \eqref{amorph1} relates the Bethe vectors to the dual Bethe vectors, cf.~\eqref{psiBC},
\begin{gather*}
\C_{a,b}(\bu;\bv) = \psi(\B_{a,b}(\bu;\bv)) = \psi(\Omega) \psi(B_{a,b}(\bu;\bv)) = \Omega^\dagger \psi(B_{a,b}(\bu;\bv)).
\end{gather*}
Due to the composition rule \eqref{keyCP} for the antimorphism $\psi$ with the coproducts~\eqref{Coproduct} and~\eqref{OpCopr}, we obtain for the composite model
\begin{gather*}
\C_{a,b}(\bu;\bv) = \Omega^{\dagger(1)}\Omega^{\dagger(2)} \left[ \Delta\circ \psi( B_{a,b}(\bu;\bv) ) \right] = \Omega^{\dagger(1)}\Omega^{\dagger(2)} \left[ (\psi\otimes \psi)\circ \Delta' (B_{a,b}(\bu;\bv)) \right] \\
= \Omega^{\dagger(1)} \Omega^{\dagger(2)} \left[ (\psi\otimes \psi)
 \sum r_1^{(1)}(\bu_{\I}) r_3^{(2)}(\bv_{\II}) \frac{f(\bu_{\II}, \bu_{\I})g(\bv_{\I},\bv_{\II})}{f(\bv_{\II},\bu_{\I})}
 B^{(1)}_{a_2,b_2}(\bu_{\II};\bv_{\II}) B^{(2)}_{a_1,b_1}(\bu_{\I};\bv_{\I}) \right].
\end{gather*}
We stress that the factors $ B_{a_2,b_2}^{(2)} (\bu_{\II};\bv_{\II}) B_{a_1,b_1}^{(1)} (\bu_{\I};\bv_{\I})$ are changed to $ B_{a_2,b_2}^{(1)} (\bu_{\II};\bv_{\II}) B_{a_1,b_1}^{(2)} (\bu_{\I};\bv_{\I})$ in comparison with Theorem~\ref{theorB}. In the same way, the functions $r_1^{(2)} (\bu_{\I}) r_3^{(1)} (\bv_{\II})$ are changed to $r_1^{(1)} (\bu_{\I}) r_3^{(2)} (\bv_{\II})$. This is due to the use of the opposite coproduct~\eqref{OpCopr} instead of the standard one~\eqref{Coproduct}. After the application of $\psi\otimes \psi$
\begin{gather*}
\C_{a,b}(\bu;\bv)= \sum r_1^{(1)}(\bu_{\I}) r_3^{(2)}(\bv_{\II}) \frac{f(\bu_{\II}, \bu_{\I})g(\bv_{\I},\bv_{\II})}{f(\bv_{\II},\bu_{\I})} \C^{(1)}_{a_2,b_2}(\bu_{\II};\bv_{\II}) \C^{(2)}_{a_1,b_1}(\bu_{\I};\bv_{\I}).
\end{gather*}
Renaming \looseness=1 the sets of variables as $\bu_{\I}\leftrightarrow \bu_{\II}$, $\bv_{\I}\leftrightarrow \bv_{\II}$, we arrive at the statement of the theo\-rem.
\end{proof}

\section[Action of $T_{13}(z)$ and $T_{23}(z)$ on $\mathcal{B}_{a,b}(\bu;\bv)$]{Action of $\boldsymbol{T_{13}(z)}$ and $\boldsymbol{T_{23}(z)}$ on $\boldsymbol{\mathcal{B}_{a,b}(\bu;\bv)}$}\label{sec:Action}

We aim to prove that the supervectors \eqref{Conjecture} satisfy
\begin{gather}\label{T13Bcal}
\frac{T_{13}(z)}{\lambda_{2}(z)h(\bv,z)} \mathcal{B}_{a-1,b-1}(\bu;\bv) = \mathcal{B}_{a,b}(\bar \eta; \bar \xi),
\end{gather}
where we introduce new sets of spectral parameters $\bar \eta=\{z,\bar u\}$ and $\bar \xi =\{z,\bar v\}$. Equation~\eqref{T13Bcal} is just one of the properties satisf\/ied by the Bethe vector $\B_{a-1,b-1}(\bar u;\bar v)$, as remarked in Appendix~\ref{app:ActionMon}. The strategy of the proof is simple. We investigate both sides of~\eqref{T13Bcal} separately and then show that they coincide. To this end, we use the known formulas for the action of the monodromy matrix elements on the Bethe vectors listed in Appendix~\ref{app:ActionMon}.

The right-hand side of \eqref{T13Bcal} has the form
\begin{gather*}
\mathcal{B}_{a,b}(\bar \eta; \bar \xi) = \sum r_1^{(2)}(\bar\eta_{\I}) r_3^{(1)}(\bar\xi_{\II})
 \frac{f(\bar\eta_{\II},\bar\eta_{\I})g(\bar\xi_{\I},\bar\xi_{\II})}{f(\bar\xi_{\II},\bar\eta_{\I})}
\B^{(2)} (\bar\eta_{\II};\bar\xi_{\II}) \B^{(1)} (\bar\eta_{\I};\bar\xi_{\I}),
\end{gather*}
where we just used def\/inition \eqref{Conjecture}. From the analysis how the parameter $z$ can enter the subsets $\bar\eta_{\I}$, $\bar\xi_{\I}$, $\bar\eta_{\II}$, $\bar\xi_{\II}$ we obtain three cases:
\begin{alignat*}{6}
&(i) \quad && \bar\eta_{\I} = \{z,\bu_{\I}\},\qquad && \bar\xi_{\I} = \{z,\bv_{\I}\}, \qquad && \bar\eta_{\II} = \bu_{\II},\qquad && \bar\xi_{\II} = \bv_{\II}, & \\
&(ii) \quad && \bar\eta_{\I} = \bu_{\I}, \qquad && \bar\xi_{\I} = \bv_{\I}, \qquad && \bar\eta_{\II} = \{z,\bu_{\II}\},\qquad && \bar\xi_{\II} = \{z,\bv_{\II}\}, & \\
&(iii) \quad && \bar\eta_{\I} = \bu_{\I}, \qquad && \bar\xi_{\I} = \{z,\bv_{\I}\}, \qquad && \bar\eta_{\II} = \{z,\bu_{\II}\}, \qquad && \bar\xi_{\II} = \bv_{\II}.&
\end{alignat*}
The case, where $z\in\bar\xi_{\II}$ and $z\in\bar\eta_{\I}$, gives a vanishing contribution because of the function $f(\bar\xi_{\II},\bar\eta_{\I})$ in the denominator. The vector $\mathcal{B}_{a,b}(\bar \eta; \bar \xi)$ is thus composed of three parts with dif\/ferent structure
\begin{gather*}
\mathcal{B}_{a,b}(\bar \eta; \bar \xi) = A_1 + A_2 + A_3,
\end{gather*}
where
\begin{gather*}
A_1 = \sum r_1^{(2)}(z) r_1^{(2)}(\bu_{\I}) r_3^{(1)}(\bv_{\II}) \frac{f(\bu_{\II},z) f(\bu_{\II},\bu_{\I}) g(\bv_{\I},\bv_{\II}) g(z,\bv_{\II})}{f(\bv_{\II},\bu_{\I}) f(\bv_{\II},z)} \\
\hphantom{A_1 =}{} \times \B^{(2)}(\bu_{\II};\bv_{\II}) \B^{(1)}(\{z,\bu_{\I}\};\{z,\bv_{\I}\}), \\
A_2 = \sum r_1^{(2)}(\bu_{\I}) r_3^{(1)}(z) r_3^{(1)}(\bv_{\II}) \frac{f(\bu_{\II},\bu_{\I}) g(\bv_{\I},z) g(\bv_{\I},\bv_{\II})}{f(\bv_{\II},\bu_{\I})} \B^{(2)}(\{z,\bu_{\II}\};\{z,\bv_{\II}\}) \B^{(1)}(\bu_{\I};\bv_{\I}), \\
A_3 = \sum r_1^{(2)}(\bu_{\I}) r_3^{(1)}(\bv_{\II}) \frac{f(z,\bu_{\I}) f(\bu_{\II},\bu_{\I}) g(z,\bv_{\II}) g(\bv_{\I},\bv_{\II})}{f(\bv_{\II},\bu_{\I})} \B^{(2)}(\{z,\bu_{\II}\};\bv_{\II}) \B^{(1)}(\bu_{\I};\{z,\bv_{\I}\}).
\end{gather*}

As the supervector $\mathcal{B}_{a-1,b-1}(\bu;\bv)$ belongs to $\mathcal{H}^{(1)}\otimes \mathcal{H}^{(2)}$, the action of $T_{13}(z)$ on it is def\/ined via a coproduct. We use the standard one~\eqref{Coproduct}
\begin{gather*}
\Delta(T_{13}(z)) = T_{13}^{(1)}(z)T_{11}^{(2)}(z) + T_{23}^{(1)}(z)T_{12}^{(2)}(z) + T_{33}^{(1)}(z)T_{13}^{(2)}(z).
\end{gather*}
Hence, the left-hand side of \eqref{T13Bcal} decomposes into three parts
\begin{gather}
 \frac{T_{13}(z)}{\lambda_{2}(z)h(\bar v,z)} \mathcal{B}_{a-1,b-1}(\bar u;\bar v) = C_1 +C_2 + C_3\nonumber \\
 \qquad {} = \left(\frac{ T_{13}^{(1)}(z)T_{11}^{(2)}(z)}{\lambda_2(z)h(\bar v,z)} + \frac{ T_{23}^{(1)}(z)T_{12}^{(2)}(z)}{\lambda_2(z)h(\bar v,z)} + \frac{ T_{33}^{(1)}(z)T_{13}^{(2)}(z)}{\lambda_2(z)h(\bar v,z)} \right) \mathcal{B}_{a-1,b-1}(\bar u;\bar v).\label{ch2}
\end{gather}
The specif\/ic parts $C_k$, $k=1,2,3$, can be written as
\begin{gather*}
 C_k = \sum (-1)^{(1+[k])\cdot b_2} r_1^{(2)}(\bar u_{\I}) r_3^{(1)}(\bar v_{\II}) \frac{f(\bar u_{\II}, \bar u_{\I})g(\bar v_{\I},\bar v_{\II})}{f(\bar v_{\II},\bar u_{\I})} \\
\hphantom{C_k =}{} \times \frac{T_{1k}^{(2)}(z)}{\lambda_2^{(2)}(z)h(\bar v_{\II},z)} \mathbb{B}^{(2)}(\bar u_{\II};\bar v_{\II})
 \frac{T_{k3}^{(1)}(z)}{\lambda_2^{(1)}(z)h(\bar v_{\I},z)}\mathbb{B}^{(1)}(\bar u_{\I};\bar v_{\I}),
\end{gather*}
where $b_2=\# \bv_{\II}$. The sign factor in $C_k$ appears because of the oddness of the monodromy matrix element~$T_{i3}(z)$, $i=1,2$. It is absorbed during the calculations by the antisymmetric functions $g(u,v)$.

Using formulas for the action of the monodromy matrix elements on the Bethe vectors listed in Appendix~\ref{app:ActionMon}, we obtain the explicit forms of~$C_k$:
\begin{gather*}
 C_1 = \sum_{\substack{ \bar u\Rightarrow \{ \bar u_{\I},\bar u_{\II}\} \\ \bar v\Rightarrow \{\bar v_{\I},\bar v_{\II}\}}}
r_1^{(2)}(\bar u_{\I}) r_3^{(1)}(\bar v_{\II}) \frac{f(\bar u_{\II}, \bar u_{\I})g(\bar v_{\I},\bar v_{\II}) }{f(\bar v_{\II},\bar u_{\I}) }
\Bigg\{ r_1^{(2)}(z) \frac{f(\bar u_{\II},z) g(z, \bar v_{\II})}{f(\bar v_{\II},z)} \mathbb{B}^{(2)}(\bar u_{\II}; \bar v_{\II}) \\
\hphantom{C_1=}{} +\sum_{\bar u_{\II}\Rightarrow \{u_{\is},\bar u_{\ii}\}} r_1^{(2)}(u_{\is}) \frac{f(\bar u_{\ii},u_{\is}) g(z,u_{\is})g(z,\bar v_{\II})}{f(\bar v_{\II},u_{\is})} \mathbb{B}^{(2)}(\{z,\bar u_{\ii}\};\bar v_{\II}) \\
\hphantom{C_1=}{}
 + \!\sum_{\substack{\bar u_{\II}\Rightarrow \{u_{\is},\bar u_{\ii}\} \\ \bar v_{\II}\Rightarrow \{v_{\is},\bar v_{\ii}\}}}\! r_1^{(2)}(u_{\is}) \frac{f(\bar u_{\ii},u_{\is}) g(v_{\is},z) g(v_{\is},\bar v_{\ii})}{f(\bar v_{\ii},u_{\is}) h(v_{\is},z) h(v_{\is},u_{\is})} \mathbb{B}^{(2)}(\{z,\bar u_{\ii}\};\{z,\bar v_{\ii}\}) \Bigg\}
 \mathbb{B}^{(1)}(\{z,\bar u_{\I}\}; \{z,\bar v_{\I}\}),\!
\\
 C_2 = \sum_{\substack{ \bar u\Rightarrow \{ \bar u_{\I},\bar u_{\II}\} \\ \bar v\Rightarrow \{\bar v_{\I},\bar v_{\II}\}}}
 r_1^{(2)}(\bar u_{\I}) r_3^{(1)}(\bar v_{\II}) \frac{f(\bar u_{\II}, \bar u_{\I})g(\bar v_{\I},\bar v_{\II})}{f(\bar v_{\II},\bar u_{\I})} \\
\hphantom{C_2=}{}
\times \Bigg\{ g(z,\bar v_{\II}) \mathbb{B}^{(2)}(\{z,\bar u_{\II}\};\bar v_{\II}) + \sum_{\bar v_{\II}\Rightarrow \{v_{\is},\bar v_{\ii}\}} \frac{g(v_{\is},z) g(v_{\is},\bar v_{\ii})}{h(v_{\is},z)} \mathbb{B}^{(2)}(\{z,\bar u_{\II}\};\{z,\bar v_{\ii}\}) \Bigg\} \\
\hphantom{C_2=}{}
\times \Bigg\{ f(z,\bar u_{\I}) \mathbb{B}^{(1)}(\bar u_{\I}; \{z,\bar v_{\I}\}) + \sum_{\bar u_{\I}\Rightarrow \{u_{\is},\bar u_{\ii}\}} g(u_{\is},z) f(u_{\is},\bar u_{\ii}) \mathbb{B}^{(1)}(\{z,\bar u_{\ii}\};\{z,\bar v_{\I}\}) \Bigg\},\\
 C_3 = \sum_{\substack{ \bar u\Rightarrow \{ \bar u_{\I},\bar u_{\II}\} \\ \bar v\Rightarrow \{\bar v_{\I},\bar v_{\II}\}}} r_1^{(2)}(\bar u_{\I}) r_3^{(1)}(\bar v_{\II}) \frac{f(\bar u_{\II}, \bar u_{\I})g(\bar v_{\I},\bar v_{\II})}{f(\bar v_{\II},\bar u_{\I})} \mathbb{B}^{(2)}(\{z,\bar u_{\II}\};\{z,\bar v_{\II}\}) \\
\hphantom{C_3=}{}
 \times \Bigg\{ r_3^{(1)}(z) g(\bar v_{\I},z) \mathbb{B}^{(1)}(\bar u_{\I};\bar v_{\I}) + \!\sum_{\bar v_{\I}\Rightarrow \{v_{\is},\bar v_{\ii}\}} \! r_3^{(1)}(v_{\is}) \frac{f(z,\bar u_{\I}) g(z,v_{\is}) g(\bar v_{\ii},v_{\is})}{h(v_{\is},z)f(v_{\is},\bar u_{\I})} \B^{(1)}(\bar u_{\I}; \{z,\bar v_{\ii}\}) \\
\hphantom{C_3=}{}
 + \sum_{\substack{\bar u_{\I}\Rightarrow \{u_{\is},\bar u_{\ii}\} \\ \bar v_{\I}\Rightarrow \{v_{\is},\bar v_{\ii}\}}} r_3^{(1)}(v_{\is}) \frac{g(u_{\is},z)f(u_{\is},\bar u_{\ii})g(z,v_{\is})g(\bar v_{\ii},v_{\is}) }{h(v_{\is},u_{\is}) f(v_{\is},z)f(v_{\is},\bar u_{\ii}) } \B^{(1)}(\{z,\bar u_{\ii}\};\{z,\bar v_{\ii}\}) \Bigg\}.
\end{gather*}
The summations are performed over all possible partitions of the sets $\bu$, $\bv$ of the original Bethe parameters into their subsets $\bar u\Rightarrow \{ \bar u_{\I},\bar u_{\II}\}$, $\bar v\Rightarrow \{\bar v_{\I},\bar v_{\II}\}$. Some of these subsets are divided into additional subsets, e.g., $\bar u_{\II}\Rightarrow \{u_{\is},\bar u_{\ii}\}$ where $\# u_{\is}=1$, and the summation is performed again over all such partitions. The same for the other additional divisions of $\bar v_{\II}$, $\bar u_{\I}$, $\bar v_{\I}$.

We can moreover see that the sum over partitions in $C_1$ involving the product of Bethe vectors $\mathbb{B}^{(2)}(\bar u_{\II}; \bar v_{\II})\mathbb{B}^{(1)}(\{z,\bar u_{\I}\}; \{z,\bar v_{\I}\})$ coincides with the term $A_1$. The sum over partitions involving $\mathbb{B}^{(2)}(\{z,\bar u_{\II}\};\bar v_{\II}) \mathbb{B}^{(1)}(\bar u_{\I};\{z,\bar v_{\I}\})$ in $C_2$ coincides with $A_3$. Similarly, the sum over partitions in~$C_3$ containing $\mathbb{B}^{(2)}(\{z,\bar u_{\II}\};\{z,\bar v_{\II}\})\mathbb{B}^{(1)}(\bar u_{\I};\bar v_{\I})$ coincides with~$A_2$. The remaining terms of~$C_1$,~$C_2$,~$C_3$ cancel mutually, as we show below.

There are two remaining terms containing the product of Bethe vectors of this type
\begin{gather*}
\B^{(2)}(\{z,\bu'\};\{z,\bv'\}) \B^{(1)}(\bu'';\{z,\bv''\}),
\end{gather*}
where primes mean any subset of $\bu$ or $\bv$. Namely, the f\/irst term comes from $C_2$
\begin{gather*}
C_{2,3} = \sum_{\substack{ \bar u\Rightarrow \{ \bar u_{\I},\bar u_{\II}\} \\ \bar v\Rightarrow \{\bar v_{\I},v_{\is},\bar v_{\ii}\}}} r_1^{(2)}(\bar u_{\I}) r_3^{(1)}(v_{\is}) r_3^{(1)}(\bar v_{\ii})
\frac{f(\bu_{\II},\bu_{\I}) g(\bv_{\I},\bv_{\ii}) g(\bv_{\I},v_{\is}) g(v_{\is},\bv_{\ii}) f(z,\bu_{\I}) }{ f(\bv_{\ii},\bu_{\I}) f(v_{\is},\bu_{\I}) h(v_{\is},z)} \\
\hphantom{C_{2,3} =}{}
\times g(v_{\is},z) \mathbb{B}^{(2)}(\{z,\bu_{\II}\};\{z,\bv_{\ii}\}) \mathbb{B}^{(1)}(\bu_{\I}; \{z,\bv_{\I}\})
\end{gather*}
and the second term from $C_3$
\begin{gather}
C_{3,2} = \sum_{\substack{ \bar u\Rightarrow \{ \bar u_{\I},\bar u_{\II}\} \\ \bar v\Rightarrow \{v_i,\bv_{\ii},\bv_{\II}\}}} r_1^{(2)}(\bar u_{\I}) r_3^{(1)}(v_{\is}) r_3^{(1)}(\bar v_{\II})
 \frac{f(\bar u_{\II}, \bar u_{\I})g(\bar v_{\ii},\bar v_{\II}) g(v_{\is},\bv_{\II}) g(\bar v_{\ii},v_{\is}) f(z,\bar u_{\I})}{f(\bar v_{\II},\bar u_{\I}) f(v_{\is},\bar u_{\I}) h(v_{\is},z)} \nonumber\\
\hphantom{C_{3,2} =}{}
\times g(z,v_{\is}) \B^{(2)}(\{z,\bu_{\II}\};\{z,\bv_{\II}\}) \B^{(1)}(\bu_{\I}; \{z,\bv_{\ii}\}).\label{C32}
\end{gather}
Renaming the sets $\bv_{\II} \rightarrow \bv_{\ii}$ and $\bv_{\ii} \rightarrow \bv_{\I}$ in \eqref{C32}, we obtain
\begin{gather*}
C_{3,2} = \sum_{\substack{ \bar u\Rightarrow \{ \bar u_{\I},\bar u_{\II}\} \\ \bar v\Rightarrow \{v_{\is},\bv_{\I},\bar v_{\ii}\}}} r_1^{(2)}(\bar u_{\I}) r_3^{(1)}(v_{\is}) r_3^{(1)}(\bar v_{\ii})
\frac{f(\bu_{\II},\bu_{\I}) g(\bv_{\I},\bv_{\ii}) g(v_{\is},\bv_{\ii}) g(\bv_{\I},v_{\is}) f(z,\bu_{\I}) }{ f(\bv_{\ii},\bu_{\I}) f(v_{\is},\bu_{\I}) h(v_{\is},z)} \\
\hphantom{C_{3,2} =}{}
\times g(z,v_{\is}) \mathbb{B}^{(2)}(\{z,\bu_{\II}\};\{z,\bv_{\ii}\}) \mathbb{B}^{(1)}(\bu_{\I};\{z,\bv_{\I}\}).
\end{gather*}
Due to the antisymmetry of the function $g(z,v_{\is})$, we see that $C_{2,3}+C_{3,2}=0$.

There are two terms containing the product of Bethe vectors of the type
\begin{gather*}
\B^{(2)}(\{z,\bu'\}; \bv') \B^{(1)}(\{z,\bu''\};\{z,\bv''\}).
\end{gather*}
One such term is contained in $C_1$ and one in $C_2$. One can prove that their sum vanishes by similar argumentation as above.

The remaining three terms contain Bethe vectors of the type
\begin{gather*}
 \B^{(2)}(\{z,\bu'\}; \{z,\bv'\}) \B^{(1)}(\{z,\bu''\};\{z,\bv''\}).
\end{gather*}
They are the following:
\begin{gather}\label{C13}
C_{1,3} = \sum_{\substack{ \bar u\Rightarrow \{ \bar u_{\I},u_{\is},\bar u_{\ii}\} \\ \bar v\Rightarrow \{\bar v_{\I},v_{\is},\bar v_{\ii}\}}}
r_1^{(2)}(\bar u_{\I}) r_1^{(2)}(u_{\is}) r_3^{(1)}(\bar v_{\ii}) r_3^{(1)}(v_{\is})
\frac{f(u_{\is},\bu_{\I}) f(\bu_{\ii},\bu_{\I}) f(\bar u_{\ii},u_{\is}) }{f(\bv_{\ii},\bu_{\I}) f(v_{\is},\bu_{\I}) f(\bar v_{\ii},u_{\is}) f(v_{\is},u_{\is})} \\
\hphantom{C_{1,3} =}{} \times \frac{g(\bv_{\I},\bv_{\ii}) g(\bv_{\I},v_{\is}) g(v_{\is},\bar v_{\ii})}{h(v_{\is},z)} g(v_{\is},z) g(v_{\is},u_{\is})
\B^{(2)}(\{z,\bar u_{\ii}\};\{z,\bar v_{\ii}\}) \B^{(1)}(\{z,\bar u_{\I}\}; \{z,\bar v_{\I}\}),\nonumber\\
C_{2,4} = \sum_{\substack{ \bar u\Rightarrow \{ u_{\is}, \bu_{\ii},\bar u_{\II}\} \\ \bar v\Rightarrow \{\bar v_{\I},v_{\is}, \bv_{\ii}\}}} r_1^{(2)}(\bu_{\ii}) r_1^{(2)}(u_{\is}) r_3^{(1)}(\bv_{\ii}) r_3^{(1)}(v_{\is})
\frac{f(\bu_{\II},\bu_{\ii}) f(\bu_{\II},u_{\is}) f(u_{\is},\bar u_{\ii}) }{ f(\bv_{\ii},\bu_{\ii}) f(\bv_{\ii},u_{\is}) f(v_{\is},\bu_{\ii}) f(v_{\is},u_{\is}) } \label{C24}\\
\hphantom{C_{2,4} =}{} \times \frac{g(\bv_{\I},\bv_{\ii}) g(\bv_{\I},v_{\is}) g(v_{\is},\bar v_{\ii})}{h(v_{\is},z)} g(v_{\is},z) g(u_{\is},z)
\B^{(2)}(\{z,\bu_{\II}\};\{z,\bv_{\ii}\}) \B^{(1)}(\{z,\bu_{\ii}\};\{z,\bv_{\I}\}),\nonumber
\\
\label{C33}
C_{3,3} = \sum_{\substack{ \bar u\Rightarrow \{ u_{\is},\bu_{\ii},\bu_{\II}\} \\ \bar v\Rightarrow \{v_{\is},\bar v_{\ii},\bar v_{\II}\}}}
r_1^{(2)}(\bu_{\ii}) r_1^{(2)}(u_{\is}) r_3^{(1)}(\bar v_{\II}) r_3^{(1)}(v_{\is})
\frac{f(\bu_{\II},\bu_{\ii}) f(\bu_{\II},u_{\is}) f(u_{\is},\bar u_{\ii})}{f(\bv_{\II},\bu_{\ii}) f(\bv_{\II},u_{\is}) f(v_{\is},\bar u_{\ii}) f(v_{\is},u_{\is})} \\
\hphantom{C_{3,3} =}{}
\times \frac{ g(\bv_{\ii},\bv_{\II}) g(v_{\is},\bv_{\II}) g(\bar v_{\ii},v_{\is})}{ h(v_{\is},z)}
 g(z,u_{\is}) g(v_{\is},u_{\is}) \B^{(2)}(\{z,\bar u_{\II}\};\{z,\bar v_{\II}\}) \B^{(1)}(\{z,\bar u_{\ii}\};\{z,\bar v_{\ii}\}).\nonumber
\end{gather}
We rename the sets as $\bu_{\II}\rightarrow \bu_{\ii}$ and $\bu_{\ii} \rightarrow \bu_{\I}$ in \eqref{C24}. In \eqref{C33}, we rename the sets as $\bu_{\II},\bv_{\II} \rightarrow \bu_{\ii},\bv_{\ii}$ and $\bu_{\ii},\bv_{\ii} \rightarrow \bu_{\I},\bv_{\I}$. Thus, we obtain all the Bethe vectors in \eqref{C13}--\eqref{C33} with the same arguments. Due to the identity
\begin{gather*}
g(v_{\is},z) g(v_{\is},u_{\is}) + g(v_{\is},z) g(u_{\is},z) + g(z,u_{\is}) g(v_{\is},u_{\is}) = 0,
\end{gather*}
 we see that $C_{1,3}+C_{2,4}+C_{3,3}=0$. Equality~\eqref{T13Bcal} is thus proved.

The action of $T_{23}(z)$ on $\mathcal{B}_{a,b-1}(\bu;\bv)$ described in \eqref{T23Bcal} can be proved in a similar manner. We again use the coproduct~\eqref{Coproduct} to def\/ine the action of $T_{23}(z)$ on the direct product of two partial Bethe vectors $\B^{(2)}_{a_2,b_2}(\bu_{\II};\bv_{\II}) \B^{(1)}_{a_1,b_1}(\bu_{\I};\bv_{\I})$ and formulas from Appendix~\ref{app:ActionMon} for the action of the monodromy matrix elements on Bethe vectors. The case of~$T_{23}(z)$ involves more lengthy calculations than that of~$T_{13}(z)$. As the reasoning is rather similar, we do not provide the details here.

\section[Composite $\mathfrak{gl}(1|2)$-invariant model]{Composite $\boldsymbol{\mathfrak{gl}(1|2)}$-invariant model} \label{sec:gl12basics}

Let us denote the RTT algebra corresponding to $\mathfrak{gl}(1|2)$ as $\wt{\mathcal{A}}$ and its elements as $\widetilde{T}_{ij}(u)$ to distinguish them from their $\mathfrak{gl}(2|1)$ equivalents. The gradation in the $\mathfrak{gl}(1|2)$ case is described by the parity function $\wt{[\ ]}$, where $\widetilde{[1]}=0$ and $\widetilde{[2]}=\widetilde{[3]}=1$. The algebraic structure is governed by the bilinear relation \eqref{Product} provided that all relevant objects are marked by tildas.

For the $\mathfrak{gl}(1|2)$-invariant models the pseudovacuum is denoted as $\wt{\Omega}$ and the corresponding eigenvalues of $\widetilde{T}_{jj}(u)$, $j=1,2,3$, are $\widetilde{\lambda}_j(u)$. The dual pseudovacuum is denoted as $\wt{\Omega}^\dagger$.

There is a relation between gradations on $\mathcal{A}$ and $\wt{\mathcal{A}}$: $[i]=\widetilde{[4-i]}+1$ $(\mathop{\rm{mod}} 2)$, $i=1,2,3$.

The superalgebras $\mathcal{A}$ and $\wt{\mathcal{A}}$ are isomorphic, as was shown in \cite{PRS16bethe}. This is due to the map
\begin{gather*}
\varphi\colon \ \begin{cases}
\mathcal{A} \rightarrow \wt{\mathcal{A}}, \\
T_{ij}(u) \rightarrow (-1)^{[i][j]+[j]+1} \widetilde{ T}_{4-j, 4-i}(u), \\
\lambda_{j}(u) \rightarrow -\widetilde{\lambda}_{4-j}(u) = \lambda_{j}(u).
\end{cases}
\end{gather*}
The ratio function $r_i(u)$, $i=1,3$, is mapped to $\wt{r}_{4-i}(u)=\wt{\lambda}_{4-i}(u)/\wt{\lambda}_{2}(u)$. Moreover $\varphi$ is a~homomorphism $\varphi(AB)=\varphi(A)\varphi(B)$. The map
\begin{gather*}
\widetilde{\varphi} \colon \
\begin{cases}
\wt{\mathcal{A}} \rightarrow \mathcal{A}, \\
\widetilde{T}_{ij}(u) \rightarrow (-1)^{\widetilde{[i]}\widetilde{[j]}+\widetilde{[j]}+1} T_{4-j,4-i}(u), \\
\wt{\lambda_{j}}(u) \rightarrow -\lambda_{4-j}(u) = \wt{\lambda}_{j}(u)
\end{cases}
\end{gather*}
is inverse to $\varphi$, i.e., $\widetilde{\varphi}\circ \varphi = \mathop{\mathrm{ id}}$ and $\varphi\circ \widetilde{\varphi} = \widetilde{ \mathop{\mathrm{ id}}}$, where $\mathop{\mathrm{ id}}$ is the identity map on $\mathcal{A}$ and $\wt{\mathop{\mathrm{ id}}}$ is the identity map on $\wt{\mathcal{A}}$.

The (dual) Bethe vectors for the $\mathfrak{gl}(1|2)$-invariant models were constructed in \cite{PRS16bethe}:
\begin{gather*}
 \widetilde{\B}_{a,b}(\bu;\bv) = (-1)^{a} \sum \frac{g(\bu_{\I},\bv_{\I}) f(\bv_{\I},\bv_{\II}) g(\bu_{\II},\bu_{\I}) h(\bv_{\I},\bv_{\I})}{\widetilde{\lambda}_2(\bu_{\II}) \widetilde{\lambda}_2(\bv) f(\bu,\bv)} \widetilde{\T}_{13}(\bv_{\I}) \widetilde{T}_{23}(\bv_{\II}) \widetilde{\T}_{12}(\bu_{\II}) \widetilde{\Omega},\\ 
 \widetilde{\C}_{a,b}(\bu;\bv) = (-1)^{\frac{a(a-1)}{2}} \sum \frac{g(\bu_{\I},\bv_{\I}) f(\bv_{\I},\bv_{\II}) g(\bu_{\II},\bu_{\I}) h(\bv_{\I},\bv_{\I})}{\widetilde{\lambda}_2(\bu_{\II}) \widetilde{\lambda}_2(\bv) f(\bu,\bv)} \widetilde{\Omega}^\dagger \widetilde{\T}_{21}(\bu_{\II}) \widetilde{T}_{32}(\bv_{\II}) \widetilde{\T}_{31}(\bv_{\I}),
\end{gather*}
where the sums go over all partitions $\bu\Rightarrow\{\bu_{\I},\bu_{\II}\}$ and $\bv\Rightarrow\{\bv_{\I},\bv_{\II}\}$ under the constraint $\#\bu_{\I}=\#\bv_{\I}$. We remind that the products of the odd operators $\wt{T}_{1i}(u)$ and $\wt{T}_{i1}(u)$, $i=2,3$, are symmetrised
\begin{gather*}
\wt{\mathbb{T}}_{1i}(\bar u) \equiv \frac{\wt{T}_{1i}(u_1)\wt{T}_{1i}(u_2)\cdots \wt{T}_{1i}(u_n)}{\prod\limits_{1\leq j<k \leq n} h(u_k,u_j)}, \qquad \wt{\mathbb{T}}_{i1}(\bar u) \equiv \frac{\wt{T}_{i1}(u_1)\wt{T}_{i1}(u_2)\cdots \wt{T}_{i1}(u_n)}{\prod\limits_{1\leq j<k \leq n} h(u_j,u_k)}.
\end{gather*}

If we assume that $\varphi(\Omega)=\widetilde{\Omega}$ and $\varphi(\Omega^\dagger)=\widetilde{\Omega}^\dagger$, it can be shown that
\begin{gather*}
\varphi(\B_{a,b}(\bu;\bv)) = \widetilde{\B}_{b,a}(\bv;\bu), \qquad \varphi(\C_{a,b}(\bu;\bv)) = \widetilde{\C}_{b,a}(\bv;\bu).
\end{gather*}

The coalgebraic structure on $\wt{\mathcal{A}}$ and $\mathcal{A}$ is related by the isomorphism $\varphi$ in the following way:
\begin{gather}\label{Dphi2}
\wt{\Delta}\circ \varphi = (\varphi\otimes \varphi) \circ \Delta',
\end{gather}
where $\Delta'$ is the opposite coproduct \eqref{OpCopr} on $\mathcal{A}$ and $\wt{\Delta}$ is the standard coproduct on~$\wt{\mathcal{A}}$
\begin{gather*}
\wt{\Delta}(\wt{T}_{ij}(u)) = - \wt{T}_{kj}(u) \otimes \wt{T}_{ik}(u).
\end{gather*}
It seems useful to incorporate the minus in the def\/inition of~$\wt{\Delta}$. This has a consequence for the pseudovacuum eigenvalues in the composite model: $\wt{\lambda}_j(u)=-\wt{\lambda}^{(1)}_j (u)\wt{\lambda}^{(2)}_j (u)$. On the other hand, the ratio functions satisfy the same relation as for $\mathfrak{gl}(2|1)$: $\wt{r}_j(u)=\wt{r}^{(1)}_j (u)\wt{r}^{(2)}_j (u)$.

\subsection{Main results} \label{sec:Results12}

From the above remarks and results of Section \ref{sec:Results21}, we can conclude about the form of Bethe vectors in the $\mathfrak{gl}(1|2)$-invariant composite model. Similarly to the case of Theorem~\ref{theorC}, the proof of the following two theorems is based on the composition rule for the isomorphism $\varphi$ with the coproducts $\wt{\Delta}$ and $\Delta'$.

\begin{Theorem} \label{theorB12}
The Bethe vectors of the full model can be expressed as the bilinear combination of the partial Bethe vectors:
\begin{gather*}
\wt{\B}_{a,b}(\bu;\bv) = \sum \wt{r}_3^{(1)}(\bv_{\II}) \wt{r}_1^{(2)}(\bu_{\I}) \frac{f(\bv_{\I}, \bv_{\II})g(\bu_{\II},\bu_{\I})}{f(\bu_{\I},\bv_{\II})} \wt{\B}^{(1)}_{a_1,b_1}(\bu_{\I}; \bv_{\I}) \wt{\B}_{a_2,b_2}^{(2)}(\bu_{\II}; \bv_{\II}) .
\end{gather*}
The summation goes over all partitions $\bar u \Rightarrow \{\bu_{\I}, \bu_{\II}\}$ and $\bv \Rightarrow \{\bv_{\I}, \bv_{\II}\}$ with no restriction. The corresponding cardinalities satisfy $a_1+a_2=a$ and $b_1+b_2=b$.
\end{Theorem}

\begin{proof} Let $\wt{B}_{ab}(\bu;\bv)$ denote the polynomial in elements of $\wt{\mathcal{A}}$ which acts on the pseudovacuum, i.e., $\wt{\B}_{ab}(\bu;\bv)=\wt{B}_{ab}(\bu;\bv)\wt{\Omega}$. We use the composition rule~\eqref{Dphi2} and the results of Theorem~\ref{theorB}
\begin{gather*}
\wt{\B}_{a,b}(\bu;\bv) = \wt{B}_{a,b}(\bu;\bv) \wt{\Omega} = \wt{\Delta}(\wt{B}_{a,b}(\bu;\bv)) \wt{\Omega}^{(1)} \wt{\Omega}^{(2)} = \big[\wt{\Delta}\circ\varphi (B_{b,a}(\bv;\bu)) \big] \wt{\Omega}^{(1)} \wt{\Omega}^{(2)} \\
 = \big[(\varphi\otimes \varphi)\circ \Delta' (B_{b,a}(\bv;\bu)) \big] \wt{\Omega}^{(1)} \wt{\Omega}^{(2)} \\
 = \left[ (\varphi\otimes \varphi) \sum r_1^{(1)}(\bv_{\I}) r_3^{(2)}(\bu_{\II}) \frac{f(\bv_{\II},\bv_{\I})g(\bu_{\I},\bu_{\II})}{f(\bu_{\II},\bv_{\I})} B^{(1)}_{b_2,a_2}(\bv_{\II};\bu_{\II}) B^{(2)}_{b_1,a_1}(\bv_{\I};\bu_{\I})\right] \wt{\Omega}^{(1)} \wt{\Omega}^{(2)} .
\end{gather*}
We stress that there was again used the opposite coproduct \eqref{OpCopr} in contrast to Theorem~\ref{theorB}. The application of the map $\varphi\otimes \varphi$ maps not only the polynomials $B^{(1)}_{b_2,a_2}(\bv_{\II};\bu_{\II}) B^{(2)}_{b_1,a_1}(\bv_{\I};\bu_{\I})$ to $\wt{B}_{a_2,b_2}^{(1)}(\bu_{\II}; \bv_{\II}) \wt{B}^{(2)}_{a_1,b_1}(\bu_{\I}; \bv_{\I})$ but also $r_1^{(1)}(\bv_{\I}) r_3^{(2)}(\bu_{\II})$ to $\wt{r}_3^{(1)}(\bv_{\I}) \wt{r}_1^{(2)}(\bu_{\II})$. Hence,
\begin{gather*}
\wt{\B}_{a,b}(\bu;\bv) = \sum \wt{r}_3^{(1)}(\bv_{\I}) \wt{r}_1^{(2)}(\bu_{\II}) \frac{f(\bv_{\II}, \bv_{\I})g(\bu_{\I},\bu_{\II})}{f(\bu_{\II},\bv_{\I})} \wt{\B}_{a_2,b_2}^{(1)}(\bu_{\II}; \bv_{\II}) \wt{\B}^{(2)}_{a_1,b_1}(\bu_{\I}; \bv_{\I}).
\end{gather*}
After renaming the sets of variables as $\bu_{\I}\leftrightarrow \bu_{\II}$ and $\bv_{\I}\leftrightarrow \bv_{\II}$, we arrive at the statement of the theorem.
\end{proof}

\begin{Theorem} \label{theorC12}
The dual Bethe vectors of the full model can be expressed as the bilinear combination of the partial Bethe vectors:
\begin{gather*}
\wt{\C}_{a,b}(\bu;\bv) = \sum \wt{r}_3^{(2)}(\bv_{\I}) \wt{r}_1^{(1)}(\bu_{\II})
 \frac{f(\bv_{\II}, \bv_{\I}) g(\bu_{\I},\bu_{\II})}{f(\bu_{\II},\bv_{\I})}
 \wt{\C}_{a_2,b_2}^{(2)}(\bu_{\II}; \bv_{\II}) \wt{\C}^{(1)}_{a_1,b_1}(\bu_{\I}; \bv_{\I}) .
\end{gather*}
The summation goes over all partitions $\bar u \Rightarrow \{\bu_{\I}, \bu_{\II}\}$ and $\bv \Rightarrow \{\bv_{\I}, \bv_{\II}\}$ with no restriction. The corresponding cardinalities satisfy $a_1+a_2=a$ and $b_1+b_2=b$.
\end{Theorem}

\begin{proof} We use the results of Theorem~\ref{theorC} and the composition rule~\eqref{Dphi2}. Everything is ana\-logous to the proof of the previous theorem
\begin{gather*}
\wt{\C}_{a,b}(\bu;\bv) = \varphi(\C_{b,a}(\bv;\bu)) = \wt{\Omega}^\dagger \varphi(C_{b,a}(\bv;\bu))
 = \wt{\Omega}^{\dagger(1)} \wt{\Omega}^{\dagger(2)} \big[ \wt{\Delta} \circ \varphi (C_{b,a}(\bv;\bu)) \big] \\
 = \wt{\Omega}^{\dagger(1)} \wt{\Omega}^{\dagger(2)} \big[ (\varphi\otimes \varphi) \circ \Delta' (C_{b,a}		(\bv;\bu)) \big] \\
 = \wt{\Omega}^{\dagger(1)} \wt{\Omega}^{\dagger(2)} \left[ (\varphi\otimes \varphi)
 	\sum r_1^{(2)} (\bv_{\II}) r_3^{(1)}(\bu_{\I})
 \frac{f(\bv_{\I},\bv_{\II}) g(\bu_{\II},\bu_{\I})}{f(\bu_{\I},\bv_{\II})}
 C_{b_1,a_1}^{(2)}(\bv_{\I};\bu_{\I}) C_{b_2,a_2}^{(1)}(\bv_{\II};\bu_{\II}) \right] \\
 = \sum \wt{r}_3^{(2)} (\bv_{\II}) \wt{r}_1^{(1)}(\bu_{\I})
 \frac{f(\bv_{\I},\bv_{\II}) g(\bu_{\II},\bu_{\I})}{f(\bu_{\I},\bv_{\II})}
 \C_{a_1,b_1}^{(2)}(\bu_{\I};\bv_{\I}) \C_{a_2,b_2}^{(1)}(\bu_{\II};\bv_{\II}).
\end{gather*}
Renaming the sets of variables as $\bu_{\I}\leftrightarrow \bu_{\II}$ and $\bv_{\I}\leftrightarrow \bv_{\II}$ leads to the statement of the theorem.
\end{proof}

\section{Conclusion}

We have obtained explicit formulas for Bethe vectors for the composite $\mathfrak{gl}(2|1)$- and $\mathfrak{gl}(1|2)$-invariant generalized quantum integrable models. The method of calculation was straightforward. We used the known action of the mo\-nodromy matrix elements on the Bethe vectors~\cite{HLPRS16multiple}. Since the RTT algebra $\mathcal{A}$ has the structure of a bialgebra, we expressed the action of the mo\-nodromy matrix elements of the complete model on the tensor product of the superspaces of the partial models using the coproduct in $\mathcal{A}$. The corresponding dual Bethe vectors were obtained using a certain antimorphism of $\mathcal{A}$. Similarly, the (dual) Bethe vectors for the $\mathfrak{gl}(1|2)$-invariant model were obtained with the help of isomorphism of the RTT algebras $\mathcal{A}$ and $\wt{\mathcal{A}}$.

The authors of \cite{PRS15gl3compBV} used apart from this approach also the coproduct property of the weight functions \cite{KP08}.

We are now prepared to calculate the form factors of the partial monodromy elements $T_{ij}^{(\ell)}(u)$ in the basis of the Bethe vectors of the complete model. They allow one to calculate the form factors of local operators depending on an internal point of the original interval $[0,L]$. The correlation functions of the local operators can be consequently investigated.

Our subsequent publication \cite{FS17} is devoted to the investigation of the form factors of the partial monodromy matrix elements and local operators using the method of zero modes \cite{PRS15zeromodes}.

\appendix

\section{Action of monodromy matrix elements on Bethe vectors}\label{app:ActionMon}

We list here some useful formulas. The summation is usually performed over all partitions of the type $\bar u \Rightarrow \{u_0,\bar u_0\}$ and $\bar v \Rightarrow \{v_0,\bar v_0\}$ where $\# u_0 = \# v_0 =1$. It also happens that the summation goes over all partitions of the type $\bar u \Rightarrow \{u_0,u_1,\bar u_2\}$ with the condition $\# u_0 = \# u_1 =1$. All formulas listed in this appendix are special cases of the results obtained in~\cite{HLPRS16multiple}.

\begin{itemize}\itemsep=0pt
\item Action of the diagonal elements:
\begin{gather*}
\frac{T_{11}(z)}{\lambda_2(z)h(\bar v,z)} \mathbb{B}_{a,b}(\bar u; \bar v) = r_1(z) \frac{f(\bar u,z)}{h(\bar v,z)} \mathbb{B}_{a,b}(\bar u; \bar v) \\
\qquad{} + \sum_{\bar u\Rightarrow \{u_0,\bar u_0 \}} r_1(u_0) \frac{g(z,u_0) f(\bar u_0,u_0) g(\bar v,z)}{f(\bar v, u_0)} \mathbb{B}_{a,b}(\{z,\bar u_0\}; \bar v) \\
\qquad{} +\sum_{\substack{\bar u \Rightarrow \{u_0, \bar u_0\} \\ \bar v \Rightarrow \{v_0, \bar v_0\} }} r_{1}(u_0) \frac{f(\bar u_0,u_0) g(z,v_0) g(\bar v_0,v_0) }{f(\bar v_0, u_0) h(v_0,z) h(v_0,u_0)} \mathbb{B}_{a,b}(\{z, \bar u_0\} ; \{z,\bar v_0\}),
\\
\frac{T_{22}(z)}{\lambda_2(z)h(\bar v,z)} \mathbb{B}_{a,b}(\bar u; \bar v) = f(z,\bar u) g(\bar v, z) \mathbb{B}_{a,b}(\bar u; \bar v) \\
 \qquad{} + \sum_{\bar v\Rightarrow\{v_0,\bar v_0\}} \frac{f(z,\bar u) g(z,v_0) g(\bar v_0,v_0)}{h(v_0,z)} \mathbb{B}_{a,b}(\bar u; \{z,\bar v_0\}) \\
 \qquad{} + \sum_{\bar u\Rightarrow\{u_0,\bar u_0\}} g(u_0,z) f(u_0,\bar u_0) g(\bar v ,z) \mathbb{B}_{a,b}(\{z,\bar u_0\}; \bar v) \\
\qquad{} + \sum_{\substack{\bar u\Rightarrow\{u_0,\bar u_0\} \\ \bar v\Rightarrow\{v_0,\bar v_0\}}} \frac{g(u_0,z)f(u_0,\bar u_0) g(z,v_0)g(\bar v_0,v_0)}{h(v_0,z)} \mathbb{B}_{a,b}(\{z,\bar u_0\}; \{z,\bar v_0\}), \\
\frac{T_{33}(z)}{\lambda_2(z)h(\bar v,z)} \mathbb{B}_{a,b}(\bar u; \bar v) = r_3(z) g(\bar v, z) \mathbb{B}_{a,b}(\bar u; \bar v) \\
\qquad{} + \sum_{\bar v\Rightarrow\{v_0,\bar v_0\}} r_3(v_0) \frac{f(z,\bar u) g(z,v_0) g(\bar v_0,v_0)}{h(v_0,z) f(v_0,\bar u)} \mathbb{B}_{a,b}(\bar u; \{z, \bar v_0\}) \\
\qquad{} + \sum_{\substack{\bar u\Rightarrow\{u_0,\bar u_0\} \\ \bar v\Rightarrow\{v_0,\bar v_0\}}} r_3(v_0) \frac{g(u_0,z)f(u_0,\bar u_0) g(z,v_0)g(\bar v_0,v_0)}{h(v_0,u_0)f(v_0,z)f(v_0,\bar u_0)} \mathbb{B}_{a,b}(\{z,\bar u_0\}; \{z, \bar v_0\}).
\end{gather*}

\item Action of the upper-triangular elements:
\begin{gather*}
\frac{T_{13}(z)}{\lambda_2(z)h(\bar v,z)} \mathbb{B}_{a,b}(\bar u; \bar v)= \mathbb{B}_{a+1,b+1}(\{z,\bar u\}; \{z,\bar v \}), \\
\frac{T_{23}(z)}{\lambda_2(z)h(\bar v,z)} \mathbb{B}_{a,b}(\bar u; \bar v) \\
\qquad{} = f(z,\bar u) \mathbb{B}_{a,b+1}(\bar u; \{ z,\bar v \}) +\sum_{\bar u\Rightarrow\{u_0,\bar u_0\}} g(u_0,z) f(u_0,\bar u_0) \mathbb{B}_{a,b+1} (\{z,\bar u_0\}, \{z,\bar v\}), \\
\frac{T_{12}(z)}{\lambda_2(z)h(\bar v,z)} \mathbb{B}_{a,b}(\bar u; \bar v) \\
\qquad{}= g(\bar v, z) \mathbb{B}_{a+1,b}(\{z,\bar u\}; \bar v) + \sum_{\bar v\Rightarrow\{v_0,\bar v_0\}} \frac{g(z,v_0) g(\bar v_0,v_0)}{h(v_0,z)} \mathbb{B}_{a+1,b}(\{z,\bar u\}; \{z, \bar v_0\}).
\end{gather*}

\item Action of the lower-triangular elements:
\begin{gather*}
\frac{T_{21}(z)}{\lambda_2(z)h(\bar v,z)} \mathbb{B}_{a,b}(\bar u; \bar v) = \sum_{\bar u\Rightarrow\{u_0,\bar u_0\}} r_1(z)\frac{f(u_0,z)f(u_0,\bar u_0)f(\bar u_0,z) g(\bar v,z)}{f(\bar v,z)h(u_0,z)} \mathbb{B}_{a-1,b}(\bar u_0; \bar v) \\
\qquad{}+ \sum_{\bar u\Rightarrow\{u_0,\bar u_0\}} r_1(u_0)\frac{f(z,u_0)f(z,\bar u_0) f(\bar u_0,u_0) g(\bar v,z)}{f(\bar v,u_0)h(z,u_0)} \mathbb{B}_{a-1,b}(\bar u_0; \bar v) \\
\qquad{} + \sum_{\bar u \Rightarrow \{u_0,u_1,\bar u_2\}} r_1(u_0) \frac{f(u_1,u_0)f(u_1,\bar u_2)f(u_1,z)f(\bar u_2,u_0)f(z,u_0)g(\bar v,z)}{f(\bar v,u_0)h(u_1,z)h(z,u_0)} \\
\qquad{} \times \mathbb{B}_{a-1,b}(\{z,\bar u_2\}; \bar v) \\
\qquad{} + \sum_{\substack{\bar u\Rightarrow\{u_0,\bar u_0\} \\ \bar v\Rightarrow\{v_0,\bar v_0\}}} r_1(u_0) \frac{f(z,\bar u_0) f(\bar u_0,u_0) g(z,v_0)g(\bar v_0,v_0)}{h(v_0,z) f(\bar v_0,u_0) h(v_0,u_0)} \mathbb{B}_{a-1,b}(\bar u_0; \{z,\bar v_0\}) \\
\qquad{}+ \sum_{\substack{\bar u\Rightarrow\{u_0,u_1,\bar u_2\} \\ \bar v\Rightarrow\{v_0,\bar v_0\}}} r_1(u_0)\frac{f(u_1,u_0)g(u_1,z)f(u_1,\bar u_2)f(\bar u_2,u_0)g(z,v_0)g(\bar v_0,v_0)}{h(v_0,z)f(\bar v_0,u_0)h(v_0,u_0)} \\
 \qquad{} \times \mathbb{B}_{a-1,b}(\{z,\bar u_2\}; \{z,\bar v_0\}).
\end{gather*}
\end{itemize}

\subsection*{Acknowledgements}

The author wants to express his gratitude to N.A.~Slavnov for the proposal to investigate this topic and discussions. He thanks also to S.~Pakuliak for discussions and to A.P.~Isaev and \v{C}.~Burd\'{i}k for their support. The work of the author has been supported by the Grant Agency of the Czech Technical University in Prague, grant No.~SGS15/215/OHK4/3T/14, and by the Grant of the Plenipotentiary of the Czech Republic at JINR, Dubna.

\pdfbookmark[1]{References}{ref}
\LastPageEnding

\end{document}